\documentclass[twoside]{article}

\usepackage[accepted]{aistats2021}
\usepackage{bm, amssymb, amsthm, graphicx, booktabs}

\newtheorem{defi}{Definition}
\newtheorem{theorem}{Theorem}
\newtheorem{proposition}{Proposition}

\newcommand{\argmax}{\mathop{\rm arg~max}\limits}
\newcommand{\argmin}{\mathop{\rm arg~min}\limits}

\begin{document}

%

%

\twocolumn[

\aistatstitle{Bayesian Model Averaging for Causality Estimation and \\its Approximation based on Gaussian Scale Mixture Distributions}

\aistatsauthor{Shunsuke Horii}

\aistatsaddress{ Waseda University } ]

\begin{abstract}
  In the estimation of the causal effect under linear Structural Causal Models (SCMs), it is common practice to first identify the causal structure, estimate the probability distributions, and then calculate the causal effect.
However, if the goal is to estimate the causal effect, it is not necessary to fix a single causal structure or probability distributions.
In this paper, we first show from a Bayesian perspective that it is Bayes optimal to weight (average) the causal effects estimated under each model rather than estimating the causal effect under a fixed single model.
This idea is also known as Bayesian model averaging.
Although the Bayesian model averaging is optimal, as the number of candidate models increases, the weighting calculations become computationally hard.
We develop an approximation to the Bayes optimal estimator by using Gaussian scale mixture distributions.
\end{abstract}

\section{Introduction}

\label{introduction}

Research on causal inference to examine the magnitude of the effect of the treatment variable on the outcome variable is one of the important tasks in data science.
Fisher's randomized controlled trial is one of the most important methods to examine the causal effect and is often considered as the gold standard for causal inference \cite{fisher1951design}.
However, complete randomization is often impracticable due to cost or ethical reasons, demanding methods for identifying causal effects for non-experimental observational studies.

To identify causal effects in observational studies, it is necessary to make some assumptions on the data generating process.
There are various interrelated approaches to statistically estimate causal effects, including those based on propensity scores \cite{guo2014propensity}, those based on instrumental variables \cite{angrist1996identification}, and those based on Structural Causal Model (SCM) \cite{pearl2000causality}.
In this study, we focus on the causality estimation based on SCM.
In SCM, we are interested in examining the causal effect of a treatment variable $X$ on an outcome variable $Y$, which is written as an experimental distribution $p(y|\mbox{do}(X=x))$, where $\mbox{do}(X=x)$ means that $X$ is fixed to $x$ by an intervention.
The causal effect $p(y|\mbox{do}(X=x))$ is calculated based on the knowledge on a causal graph $G$, which describes a coarse relationship among the observable variables, and a distribution $P(X, Y, Z_{1},\ldots,Z_{c})$, where $Z_{1},\ldots,Z_{c}$ are observable variables (covariates) other than $X$ and $Y$.

In general, we do not know the true causal graph $G$, so there are various methods for estimating the causal graph from the data \cite{spirtes1991algorithm, cooper1992bayesian, heckerman1995learning, shimizu2006linear}.
Furthermore, to calculate the causal effect, it has to estimate the conditional distributions among the variables.
Therefore, a general way to estimate the causal effect consits of the following steps.
\begin{enumerate}
    \item Estimate the causal graph from the data
    \item Estimate the conditional distributions among the variables from the data
    \item Calculate the causal effect
\end{enumerate}
However, if the goal is to estimate the causal effect, it is not necessary to fix a single causal graph or distributions.
We first derive the Bayesian optimal estimator when the problem is to estimate the mean causal effect under the squared error loss.
The Bayes optimal solution turns out to be the estimator that weights the mean causal effects estimated under each model, which is also known as Bayesian model averaging \cite{hoeting1999bayesian}.
In the literature of causal inference, Bayesian model averaging is applied for propensity score analysis \cite{kaplan2014bayesian}, however, our study is an attempt to apply it for the causal inference based on SCM.
The basic idea behind our proposal is to consider the causal graph and the parameters of the distributions as parameters to be marginalized in estimating the causal effect.
Although the model setting is different, a similar idea can be found in \cite{rubin1978bayesian}, where the values of the unobserved potential outcomes are treated as parameters to be marginalized.

Although the Bayesian model averaging is optimal, as the number of candidate models increases, the weighting calculations become computationally hard.
In this paper, we develop an approximation algorithm for the case where the following assumptions hold.
\begin{enumerate}
    \item We know the set of possible directed edges, including information on the directions of the edges.
    \item We do not know whether each element of the above set of possible directed edges exists. We only know the prior probability that each directed edge exists.
\end{enumerate}
The above assumptions make sense in some practical applications.
For example, consider the case where there are $m$ variables $X_{1},\ldots,X_{m}$ and only the antecedent relations, such as nodes with higher index numbers do not precede nodes with lower index numbers, to the causal relations among these variables are known.
This is the situation dealt with in \cite{wermuth1982graphical}.
In this case, the set of possible directed edges are given by $\left\{(i,j):i<j, 1\le i,j\le m\right\}$.

Even if we put the constraint above, the number of models grows exponentially with the number of possible directed edges.
We develop an approximation to the Bayes optimal estimator by using sparse modeling techniques.
The basic idea is to approximate a mixture distribution of Dirac delta function and Gaussian distribution with Gaussian scale mixture (GSM).
In the prediction task, it is reported that performance similar to Bayesian model averaging can be achieved by using horseshoe prior, which is a kind of GSM \cite{carvalho2009handling}.
This study shows that a similar approach is also effective in estimating causal effects.

The paper is organized as follows.
In Section 2, some preliminary materials about SCM and corresponding causality notions such as the intervention effect and the mean intervention effect are described.
In Section 3, we formulate the problem to estimate the mean intervention effect as a statistical decision problem and derive the optimal decision function under the Bayes criterion.
Section 4 describes the further assumptions for the causal graph in this study and derives an approximation algorithm for the Bayes optimal decision function under these assumptions.
Section 5 presents some experimental results to evaluate our proposals.
Finally, we give a summary in Section 6.

\section{Preliminaries}

As mentioned in the introduction, there are some approaches to statistically estimate causal effects.
Since this study follows the framework of SCM, we borrow some basic notions of SCM.
\begin{defi}
Let $G$ be a directed acyclic graph (DAG) and $V=(X_{1},X_{2},\ldots,X_{m})$ be a set of random variables that corresponds to the set of the vertices of $G$.
We sometimes write $E_{G}$ as the set of directed edges of $G$ and $G_{E}$ as the DAG whose set of edges is $E$.
The DAG $G$ is called a causal graph if it specifies the causal relationships among variables in the following form,
\begin{align}
    X_{i} = g_{i}({\rm pa}(X_{i}),\epsilon_{i}),\quad i=1,\ldots,m,\label{structural_equation}
\end{align}
where ${\rm pa}(X_{i})\subset V$ is the set of variables that have a directed edge that heads to $X_{i}$ and $\epsilon_{i}$ is an error term\footnote{Unless it causes a confusion, ${\rm pa}(X)$ is also written as ${\rm pa}(x)$. They sometimes denote the indices of the nodes.}.
In this study, we assume that $\epsilon_{i}$ follows the Gaussian distribution $\mathcal{N}(0, s_{\epsilon}^{-1})$.
The equations (\ref{structural_equation}) are called structural equations for $X_{1}, X_{2}, \ldots, X_{m}$.

When the functions $g_{i},i=1,\ldots,m$ are linear, i.e.,
\begin{align}
    X_{i}=\sum_{X_{j}\in {\rm pa}(X_{i})}\theta_{X_{j}X_{i}}X_{j}+\epsilon_{i},\quad i=1,\ldots,m,
\end{align}
the model is called linear SCM.
If there is no confusion, $\theta_{X_{i}X_{j}}$ is written as $\theta_{ij}$ for short.
\end{defi}

One may think that the parametric assumption among the variables is too strong, however, a theory built on a simple model will be the foundations for a theory in more complex models, and as we will show in later experiments, even such a simple model works for some real-world problems.

Structural equations and causal graphs express coarse causal relationships among the variables.
Given these information, we want to know the causal effect of a treatment variable $X\in V$ on an outcome variable $Y\in V$ when $X$ is fixed to a value $x$ by an external intervention.
It is mathematically defined as follows \cite{pearl2000causality}.

\begin{defi}
Let $V=\left\{X, Y, Z_{1},\ldots,Z_{c}\right\}$ be the set of vertices of a causal graph $G$.
The causal (intervention) effect on $Y$ when $X$ is fixed to $x$ by an external intervention is defined as
\begin{multline}
    p(y|{\rm do}(X=x))=\\ 
    \int\cdots\int \frac{p(x,y,z_{1},\ldots,z_{c})}{p(x|{\rm pa}(x))}{\rm d}z_{1}\ldots {\rm d}z_{c},
\end{multline}
where ${\rm do}(X=x)$ means that $X$ is fixed to $x$ by an intervention.
\end{defi}

The intervention effect $p(y|{\rm do}(X=x))$ is defined as an experimental distribution.
In this study, we focus on the estimation of their mean (expectation), the Mean Intervention Effect (MIE).
The MIE $\bar{y}_{x}$ is defined as
\begin{align}
    \bar{y}_{x}=\int y\cdot p(y|\mbox{do}(X=x)){\rm d}y.
\end{align}
When the model is linear SCM, it is known that the MIE is expressed as 
\begin{align}
    \bar{y}_{x}=\left(\sum_{l\in\mathcal{P}}\prod_{(i,j)\in l}\theta_{ij}\right)x,\label{MIE}
\end{align}
where $\mathcal{P}$ is the set of the directed paths from $X$ to $Y$ \cite{pearl2000causality}.
This is an equivalent notion of the total effect in \cite{wright1921correlation}.

We note that the main interest of the existing studies in the literature is the identifiability of the total effect.
For example, Pearl proved a valuable result that claims that we can identify the total effect if we can observe some covariates that satisfy the backdoor criterion \cite{pearl2000causality}.
On the other hand, our focus is how to estimate the total effect.
Although it would be interesting to combine our results with the existing results, we do not use those results in this paper.

\section{Bayesian estimation of the mean intervention effect}

In order to calculate the MIE (\ref{MIE}), one has to know the underlying causal graph $G$ and conditional distributions $p(x_{i}|{\rm pa}(x_{i})), i=1,\ldots,m$.
Therefore, the data analyst must either proceed with the belief that the assumed causal graph is correct, or estimate the causal graph from the data.
Although there are various methods for estimating the causal graph from the data \cite{spirtes1991algorithm, cooper1992bayesian, heckerman1995learning, shimizu2006linear}, there is a possibility that these methods output a wrong causal graph.
Furthermore, previous studies are rarely concerned about how to estimate the probability distributions $p(x_{i}|{\rm pa}(x_{i})), i=1,\ldots,m$ because their main concern is about the identifiability of the causal effects given the causal graph and probability distributions.
However, it is uncommon to assume that the probability distribution is known without knowing the data generating causal graph.
In our setting, we need to deal with the estimation of the probability distributions as well as the estimation of the causal graph.

We assume that the causal diagram $G$ is a random variable that takes its value in the set of DAGs $\mathcal{G}$ and whose prior distribution is $p(G)$.
This prior distribution represents, for example, the data analyst's confidence in each causal graph.
Since we deal with the linear SCM, given a causal graph $G$, conditional distributions $p(x_{i}|{\rm pa}(x_{i}))$ are parameterized by $\theta_{ij},(i,j)\in E_{G}$ as follows.
\begin{align}
    p(x_{i}|{\rm pa}(x_{i}))=\mathcal{N}\left(\sum_{j:(j,i)\in E_{G}}\theta_{ji}x_{j},s_{\epsilon}^{-1}\right).
\end{align}
Throughout the paper, we assume that the precision parameter $s_{\epsilon}$ of error is known.
We can extend our results to the unknown case by assuming a prior distribution for $s_{\epsilon}$.
Let $\bm{\theta}_{G}=\left(\theta_{ij}:(i,j)\in E_{G}\right)$.
We assume that the parameter $\bm{\theta}_{G}$ is also a random vector whose conditional distribution under $G$ is $p(\bm{\theta}_{G}|G)$.
Since the MIE is the function of $G$ and $\bm{\theta}_{G}$, we write it as $\bar{y}_{x}(G,\bm{\theta}_{G})$.

Let $D^{N}=\left(x_{n}, y_{n}, z_{1,n},\ldots,z_{c,n}\right)_{n=1,\ldots,N}$ be a sample of $V=(X, Y, Z_{1},\ldots,Z_{c})$ with sample size $N$. \footnote{We sometimes do not distinguish $X$ and $Y$ from other covariates $Z_{1},\ldots,Z_{c}$. In that case, we write the set of variables as $V=(X_{1},X_{2},\ldots,X_{m})$.}
We consider the problem of estimating the MIE (\ref{MIE}) given $D^{N}$.
Let $d:D^{N}\mapsto \mathbb{R}$ be a decision function that outputs an estimate of the MIE.
The loss function is the metric between the estimand and the decision function.
In this study, the squared error loss is used, i.e., 
\begin{align}
    \ell(G, \bm{\theta}_{G}, d(D^{N}))=\left(\bar{y}_{x}(G, \bm{\theta}_{G})-d(D^{N})\right)^{2}. \label{square_error}
\end{align}
In the statistical decision theory framework \cite{berger2013statistical}, the risk function and the Bayes risk function are defined as follows, respectively.
\begin{align}
    R(G, \bm{\theta}_{G},d)&={\rm E}_{D^{N}|G,\bm{\theta}_{G}}\left[\ell(G,\bm{\theta}_{G},d(D^{N}))\right],\\
    BR(d)&={\rm E}_{G}\left[{\rm E}_{\bm{\theta}_{G}|G}R(G,\bm{\theta}_{G},d)\right].
\end{align}
The Bayes optimal estimator, that minimizes the Bayes risk function, is given as follows.
\begin{theorem}
When the loss function is the squared loss (\ref{square_error}), the Bayes optimal estimator of the MIE is given by
\begin{multline}
d^{*}(D^{N})=\\
    \sum_{G\in\mathcal{G}}p(G|D^{N})\int \bar{y}_{x}(G,\bm{\theta}_{G})p(\bm{\theta}_{G}|G,D^{N}){\rm d}\bm{\theta}_{G},\label{BO_estimator}
\end{multline}
where
\begin{align}
    p(G|D^{N})&=\frac{p(D^{N}|G)p(G)}{\sum_{G\in\mathcal{G}}p(D^{N}|G)p(G)},\label{bayes_rule}\\
    p(D^{N}|G)&=\int p(D^{N}|G,\bm{\theta}_{G})p(\bm{\theta}_{G}|G){\rm d}\bm{\theta}_{G},\\
    p(\bm{\theta}_{G}|G,D^{N})&=\frac{p(D^{N}|G,\bm{\theta}_{G})p(\bm{\theta}_{G}|G)}{\int p(D^{N}|G,\bm{\theta}_{G})p(\bm{\theta}_{G}|G){\rm d}\bm{\theta}_{G}}.
\end{align}
\end{theorem}
\begin{proof}
It is known that the decision function that minimizes the loss function weighted by the posterior distribution is Bayes optimal \cite{berger2013statistical}.
That is,
\begin{align}
    d^{*}(D^{N})=\argmin_{d}\int  \left\{\left(\bar{y}_{x}(G, \bm{\theta}_{G})-d(D^{N})\right)^{2}\right.\times \nonumber\\
    \left. \sum_{G\in\mathcal{G}}p(G|D^{N}) p(\bm{\theta}_{G}|G,D^{N})\right\}{\rm d}\bm{\theta}_{G}.
\end{align}
The solution of this minimization problem is given by (\ref{BO_estimator}).
\end{proof}
Note that the Bayes optimal estimator of the intervention effect under the Kullback-Leibler loss is given in \cite{horii2019note}.

In general, numerical integration is required to calculate the integrals in (\ref{BO_estimator}) since $\bar{y}_{x}(G,\bm{\theta}_{G})$ is a nonlinear function of $\bm{\theta}_{G}$.
When the computational complexity of the numerical integration is large, we approximate (\ref{BO_estimator}) by
\begin{align}
    \tilde{d}^{*}(D^{N})&=\sum_{G\in\mathcal{G}}p(G|D^{N})\bar{y}_{x}(G,\bm{\theta}_{G}^{MAP}),\label{BO_estimator_2}\\
    \bm{\theta}_{G}^{MAP}&=\argmax_{\bm{\theta}_{G}}p(\bm{\theta}_{G}|G,D^{N}).\label{MAP_estimator}
\end{align}
This approximation is based on the property that the posterior distribution $p(\bm{\theta}_{G}|G,D^{N})$ is asymptotically concentrated around $\bm{\theta}_{G}^{MAP}$ under some appropriate conditions \cite{le2012asymptotic}.
To keep the description concise, we call the estimator (\ref{BO_estimator_2}) Bayes quasi-optimal estimator.

If the prior distribution $p(\bm{\theta}_{G}|G)$ is a product of Gaussian, namely, 
\begin{align}
    p(\bm{\theta}_{G}|G)=\prod_{(i,j)\in E_{G}}\mathcal{N}(\theta_{ij};0,\tau),\label{parameter_prior}
\end{align}
we can analytically calculate $p(G|D^{N})$ and $p(\bm{\theta}_{G}|G, D^{N})$.
See the supplementary material for the derivation.

\section{Approximate Bayes Optimal Estimator}

The problem with calculating (\ref{BO_estimator}) or (\ref{BO_estimator_2}) in practical applications is that its computational complexity is proportional to the number of candidate models.
The number of possible DAGs for a given set of variables $V=(X_{1}, X_{2},\ldots, X_{m})$ is $O\left(2^{m^{2}}\right)$.
However, it would be uncommon to know nothing at all about the structure of a graph, and we might know the causal relationships among some variables and not the rest.
Therefore, in this study, we classify the directed edges expressing the causal relationships among variables into three types.
\begin{itemize}
    \item It is known that there is a causal relationship between variables, including information on which is the cause and which is the result.
    In other words, there is a directed edge between the corresponding variables with probability 1.
    \item It is known that there is no causal relationship between the corresponding variables.
    In other words, there is directed edge between the corresponding variables with probability 0.
    \item If there is a causal relationship, it is known that which is the cause and which is the result, but it is not clear whether the causal relationship exists.
    In other words, there is a directed edge between the corresponding variables with some probability.
\end{itemize}

For the sake of simplicity, we assume that the set of the edges of the first type is empty, since they can be considered as the edges of the third type with edge existence probability set to 1.

Let $E_{\rm full}$ be a set of possible directed edges connecting nodes that may or may not have a causal relationship between the corresponding variables.
We assume that $(j,i)\notin E_{\rm full}$ if $(i,j)\in E_{\rm full}$.
The set $\mathcal{G}$ of the candidate causal graphs is defined as
\begin{align}
    \mathcal{G}=\left\{G\ :\ E_{G}\subseteq E_{\rm full}\right\}.\label{restricted_models}
\end{align}
Let $G_{\rm full}$ be the abbreviation for $G_{E_{\rm full}}$.
Figure \ref{fig:graph_example} depicts an example of $G_{\rm full}$ and $\mathcal{G}$.

\begin{figure}[t]
\vskip 0.2in
\begin{center}
\centerline{\includegraphics[width=\columnwidth]{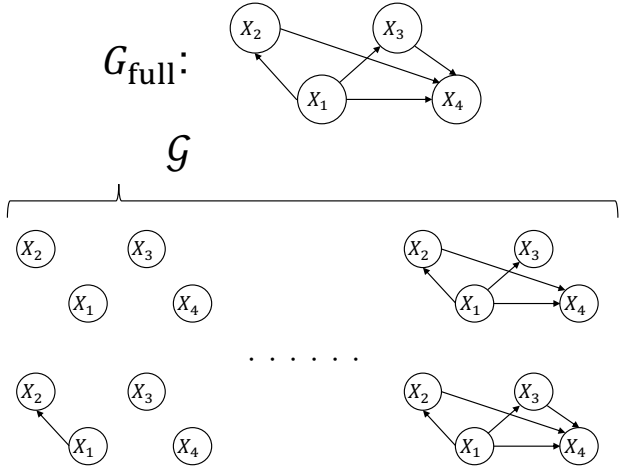}}
\caption{An example of $G_{\rm full}$ and $\mathcal{G}$. In this case, the number of candidate causal graphs $|\mathcal{G}|$ is $2^{5}=32$.}
\label{fig:graph_example}
\end{center}
\vskip -0.2in
\end{figure}

We assume that there is a causal relationship between $X_{i}$ and $X_{j}$ for $(i,j)\in E_{\rm full}$ with probability $p$, that is, a directed edge $(i,j)$ exists with probability $p$.
Then, the prior distribution $p(G)$ is given by
\begin{align}
    p(G)=p^{|E_{G}|}(1-p)^{|E_{\rm full}\setminus E_{G}|},\quad  \forall G\in \mathcal{G}.\label{model_prior}
\end{align}

Even if the set of candidate models is restricted to (\ref{restricted_models}), the number of candidate models $|\mathcal{G}|$ is $2^{|E_{\rm full}|}$, and when $|E_{\rm full}|$ is large, the calculation of (\ref{BO_estimator}) becomes infeasible.

For $G_{\rm full}$, consider the following experimental prior distribution.
\begin{align}
    &p_{ss}(\bm{\theta}_{G_{\rm full}})=\prod_{(i,j)\in E_{\rm full}}p_{ss}(\theta_{ij}),\label{spike_slab_1}\\
    &p_{ss}(\theta_{ij})=\left(1-p)\delta_{0}(\theta_{ij})+p\mathcal{N}(\theta_{ij}; 0, \tau)\right),\label{spike_slab_2}
\end{align}
where $\delta_{0}(\cdot)$ is the Dirac delta function.
Such priors are often called spike-and-slab priors \cite{ishwaran2005spike}.
Then, the following proposition holds.
\begin{proposition}
The prior distribution of $\theta_{ij}, (i,j)\in E_{\rm full}$ is the same when (\ref{spike_slab_1}) and (\ref{spike_slab_2}) are assumed for $p(\bm{\theta}_{G_{\rm full}})$ and (\ref{parameter_prior}) and (\ref{model_prior}) are assumed for $p(\bm{\theta}_{G}|G)$ and $p(G)$.
That is, 
\begin{align}
    p_{ss}(\bm{\theta}_{G_{\rm full}})=\sum_{G\in \mathcal{G}}p(\bm{\theta}_{G}|G)p(G).
\end{align}
\end{proposition}
\begin{proof}
For $(i,j)\in E_{\rm full}$,
\begin{multline}
    \sum_{G\in\mathcal{G}}p(\theta_{ij}|G)p(G)=\\
    \sum_{G:(i,j)\notin E_{G}}\delta_{0}(\theta_{ij})p(G)+\sum_{G:(i,j)\in E_{G}}\mathcal{N}(\theta_{ij};0,\tau)p(G)
\end{multline}
Then, $\sum_{G:(i,j)\in E_{G}}p(G)$ is the probability that a graph $G$ has the edge $(i,j)$ and it is given by $p$.
Similarly, $\sum_{G:(i,j)\notin E_{G}}p(G)=1-p$.
\end{proof}
When (\ref{spike_slab_1}) and (\ref{spike_slab_2}) are assumed, the Bayes optimal estimator can be rewritten as 
\begin{align}
    \int \bar{y}_{x}(G_{\rm full},\bm{\theta}_{G_{\rm full}})p_{ss}(\bm{\theta}_{G_{\rm full}}|D^{N}){\rm d}\bm{\theta}_{G_{\rm full}},\label{BO_estimator_ss}
\end{align}
where 
\begin{align}
    p_{ss}(\bm{\theta}_{G_{\rm full}}|D^{N})&=\frac{p(D^{N}|\bm{\theta}_{G_{\rm full}})p_{ss}(\bm{\theta}_{G_{\rm full}})}{\int p(D^{N}|\bm{\theta}_{G_{\rm full}})p_{ss}(\bm{\theta}_{G_{\rm full}}){\rm d}\bm{\theta}_{G_{\rm full}}}.\label{posterior_ss}
\end{align}
With the introduction of the spike-and-slab prior, the problem of calculating the summation over the candidate models has apparently disappeared, but the calculation of (\ref{BO_estimator_ss}) is still difficult because the calculation of the posterior distribution (\ref{posterior_ss}) is difficult.

We approximate $p_{ss}(\theta_{ij})$ by Gaussian scale mixture (GSM) \cite{andrews1974scale,boris2008scale}.
The probability density function (pdf) of GSM is given by
\begin{align}
    p_{gs}(\theta_{ij})=\int \mathcal{N}(\theta_{ij};0,\tau_{ij})p(\tau_{ij};\bm{\alpha})\rm{d}\tau_{ij},
\end{align}
where $\bm{\alpha}$ is the parameter of the distribution $p(\bm{\tau};\bm{\alpha})$.
Note that the values of $\tau_{ij}$ are different for different $(i,j)\in E_{\rm full}$.
If $\tau_{ij}$ is small, it means that the probability that $\theta_{ij}$ takes a value close to 0 is large, and we want to estimate them from the data.
The distribution $p(\tau_{ij};\bm{\alpha})$ is a distribution of the variance of Gaussian and it is often called mixing distribution.
It is known that GSM can express various distributions, such as Laplace distribution, Student-t distribution and horseshoe distribution, by changing the mixing distribution.
GSM is often used in the Bayesian sparse modeling literature.
See \cite{ji2008bayesian}, for example.

Even if we use GSM for the prior distribution $p(\theta_{ij})$, it is still difficult to calculate the exact posterior, however, there are various approximation algorithms to efficiently calculate the approximate posterior which are based on Expectation-Maximization (EM) algorithm \cite{figueiredo2003adaptive}, Markov-Chain Monte-Carlo (MCMC) algorithm \cite{park2008bayesian}, and Variational Bayes (VB) algorithm \cite{babacan2014bayesian}.
As an example, we give an estimation algorithm based on the variational Bayes algorithm.
See the supplementary material for the derivation of the algorithm.

Let $\bm{\theta}_{j}=\left(\theta_{ij}:(i,j)\in E_{\rm full}\right)$ and we describe an estimation algorithm for $\bm{\theta}_{j}$ since we can estimate $\bm{\theta}_{1},\ldots,\bm{\theta}_{m}$ independently.
We use the exponential distribution as the mixing distribution, that is, 
\begin{align}
    p(\tau_{ij}|\alpha_{ij})=\left\{
    \begin{array}{ll}
    \alpha_{ij}e^{-\alpha_{ij}\tau_{ij}} & \tau_{ij}\ge 0,\\
    0 & \tau_{ij}<0.
    \end{array}
    \right.
\end{align}
We have to determine or estimate the values of $\alpha_{ij}$.
We take a Bayesian hierarchical modeling approach, that is, we further assume gamma distribution for $\alpha_{ij}$ which is the conjugate distribution of the exponential distribution.
\begin{align}
    p(\alpha_{ij};\kappa,\nu)=\frac{\nu^{\kappa}}{\Gamma(\kappa)}\alpha_{ij}^{\kappa-1}e^{-\nu \alpha_{ij}},
\end{align}
where $\Gamma(\cdot)$ is the gamma function and $\kappa, \nu$ are hyper-parameters.
By setting $\kappa, \nu$ so that $p(\alpha_{ij};\kappa, \nu)$ is flat, the algorithm can estimate $\alpha_{ij}$ as well as other parameters.
The update equations for the variational Bayes algorithm are summarized as follows.
\begin{itemize}
    \item Update equation for $\bm{\theta}_{j}$
    \begin{align}
        \bar{\bm{\theta}}^{(t+1)}_{j}&=s_{\epsilon}\bm{\Sigma}^{(t+1)}_{j}\bm{X}_{j}^{T}\bm{x}_{j},\\
        \bm{\Sigma}_{j}^{(t+1)}&=\left(s_{\epsilon}\bm{X}_{j}^{T}\bm{X}_{j}+\bar{\bm{S}}_{j}^{(t)}\right)^{-1},
    \end{align}
    where 
    \begin{align}
        \bar{\bm{S}}_{j}^{(t)}={\rm diag}\left(\bar{s}^{(t)}_{j,1},\ldots, \bar{s}^{(t)}_{j,m_{j}}\right),
    \end{align}
    and ${\rm diag}(\bm{a})$ is the diagonal matrix whose diagonal elements are $\bm{a}$.
    \item Update equation for $\left\{\tau_{j,i}\right\}$
    \begin{align}
        \bar{\tau}_{j,i}^{(t+1)}&=\frac{1+\sqrt{\bar{\alpha}_{j,i}^{(t)}\left((\bar{\theta}_{j,i}^{(t+1)})^{2}+\Sigma_{j,ii}^{(t+1)}\right)}}{\bar{\alpha}_{j,i}^{(t)}},\\
        \bar{s}_{j,i}^{(t+1)}&=\sqrt{\frac{\bar{\alpha}_{j,i}^{(t)}}{(\bar{\theta}_{j,i}^{(t+1)})^{2}+\Sigma_{j,ii}^{(t+1)}}},
    \end{align}
    where $\bar{\theta}_{j,i}^{(t)}$ and $\Sigma_{j,ii}^{(t+1)}$ are the $i$-th element of $\bar{\bm{\theta}}_{j}^{(t)}$ and $(i,i)$-element of $\bm{\Sigma}_{j}^{(t+1)}$, respectively\footnote{To make the description concise, indices of the variables are replaced so that $\theta_{j,i}=\theta_{j_{i}j}$.}.
    \item Update equation for $\left\{\alpha_{j,i}\right\}$
    \begin{align}
        \bar{\alpha}_{j,i}^{(t+1)}=(\kappa+1)\left(\nu+\frac{\bar{\tau}_{j,i}^{(t+1)}}{2}\right).
    \end{align}
\end{itemize}
Starting from some initial values $\left\{\bar{s}_{j,i}^{(0)}\right\}, \left\{\bar{\alpha}_{j,i}^{(0)}\right\}$ and iterating the above algorithm until it converges, we obtain an approximation posterior distribution $q(\bm{\theta}_{j}|D^{N})=\mathcal{N}(\hat{\bm{\theta}}_{j}, \hat{\bm{\Sigma}}_{j}),j=1,\ldots, m$, where $\hat{\bm{\theta}}_{j}$ and $\hat{\bm{\Sigma}}_{j}$ are the convergence values of $\bar{\bm{\theta}}_{j}^{(t)}$ and $\bm{\Sigma}_{j}^{(t)}$, respectively.
Using these, (\ref{BO_estimator_ss}) and (\ref{BO_estimator_2}) are approximated as 
\begin{align}
    d_{VB}(D^{N})&=\int \bar{y}_{x}(G_{\rm full}, \bm{\theta}_{G_{\rm full}})\prod_{j=1}^{m}q(\bm{\theta}_{j}|D^{N}){\rm d}\bm{\theta}_{G_{\rm full}},\\
    \tilde{d}_{VB}(D^{N})&=\bar{y}_{x}\left(G_{\rm full}, \left\{\hat{\bm{\theta}}_{j}\right\}_{j=1,\ldots,m}\right),\label{VB_estimator2}
\end{align}
respectively.

The computational complexity of an iteration of the algorithm is $O(m_{j}^{3})$, which comes from the inversion of the $m_{j}\times m_{j}$ matrix.
If the dimension of the problem is too high to explicitly calculate $\bm{\Sigma}_{j}^{(t+1)}$, we have to use an approximate matrix inversion algorithm with low complexity.
In this paper, we do not deal with such high-dimensional problems.

\section{Experiments}

\subsection{Experiments on synthetic data}

To verify the effectiveness of the proposed method, we have to compare it with a conventional method.
A general method of calculating the MIE is to first estimate the graph structure $G$, estimate the parameters $\bm{\theta}_{G}$ of the conditional probability distributions, and finally calculate the MIE by (\ref{MIE}).
The K2 algorithm \cite{cooper1992bayesian} is used for the learning of the graph structure and the posterior mean is used for the estimator of $\bm{\theta}_{G}$.
The K2 algorithm requires a metric to compare two models.
The posterior probabilities of the models are used as the metric for the K2 algorithm.
It is a greedy algorithm and it adds a directed edge if the posterior of the model increases.
Another method to be compared is to estimate the MIE under the graph $G_{\rm full}$, which assumes that all potential direct edges exist.
It calculates $\bm{\theta}_{G_{\rm full}}^{MAP}$ according to (\ref{MAP_estimator}) and then computes the MIE.

Graphs in the form of Figure \ref{fig:experiment_graph} are used as $G_{\rm full}$, where $W_{1},\ldots, W_{n_{1}}$ and $Z_{1},\ldots,Z_{n_{2}}$ are the covariates.
This graph has following edges:
\begin{itemize}
    \item Edges $(W_{i},X)$ and $(W_{i},Y)$ for all $i=1,\ldots,n_{1}$
    \item Edges $(X, Z_{i})$ for all $i=1,\ldots,n_{2}$
    \item Edges $(W_{i}, Z_{j})$ for all $i=1,\ldots,n_{1}, j=1,\ldots, n_{2}$
\end{itemize}
A covariate $W_{i}$ causes a pseudo correlation between $X$ and $Y$ if the edges $(W_{i}, X)$ and $(W_{i}, Y)$ exist, while $X$ has an indirect effect on $Y$ through $Z_{i}$ if the edges $(X,Z_{i})$ and $(Z_{i},Y)$ exist.

First, we compare the proposed method with a conventional method based on the K2 algorithm and the Bayes quasi-optimal estimator for a small $G_{\rm full}$.
The proposed method is the estimator (\ref{VB_estimator2}) and the Bayes quasi-optimal estimator is the estimator (\ref{BO_estimator_2}).
The hyper-parameters $\kappa, \nu$ of the proposed method is set to $\kappa=\nu=10^{-6}$.
The number of the covariates is set to $n_{1}=n_{2}=2$ and the edge appearance probability is set to $p=0.5$.
We consider the problem of estimating the MIE (\ref{MIE}) when $x=1$.
Figure \ref{fig:experiment_Bayes} shows the squared error curves for the MIE as the functions of the sample size.
We can see that the proposed method outperforms the estimator based on the K2 algorithm and that based on the full model.
We can also see that its performance is close to that of the Bayes quasi-optimal estimator.

Then, we compare the proposed method with the conventional estimators for a large $G_{\rm full}$.
The number of the covariates is set to $n_{1}=n_{2}=30$ and the edge appearance probability is set to $p=0.3, 0.5, 0.7$.
The other settings are the same as the previous experiment.
In this case, since the number of candidate graphs $|\mathcal{G}|$ is $2^{960}$, we can not compute the Bayes quasi-optimal estimator.
Figure \ref{fig:experiment_2} shows the mean squared error curves for the MIE as the functions of the sample size.
We can see that the proposed method outperforms the K2 based estimator and full model based estimator for the all values of $p$.
We think that these experiments demonstrate that the Bayesian model averaging can be approximated with high accuracy by using GSM for the total effect inference problem.

\begin{figure}[t]
\vskip 0.2in
\begin{center}
\centerline{\includegraphics[width=\columnwidth]{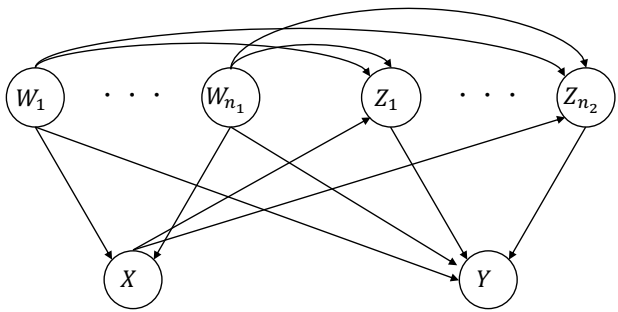}}
\caption{An example of $G_{\rm full}$ used for the experiment. $W_{1},\ldots,W_{n_{1}}$ would cause pseudo correlations between $X$ and $Y$ and $X$ would have indirect effects on $Y$ through $Z_{1},\ldots,Z_{n_{2}}$.}
\label{fig:experiment_graph}
\end{center}
\vskip -0.2in
\end{figure}

\begin{figure}[t]
\vskip 0.2in
\begin{center}
\centerline{\includegraphics[width=\columnwidth]{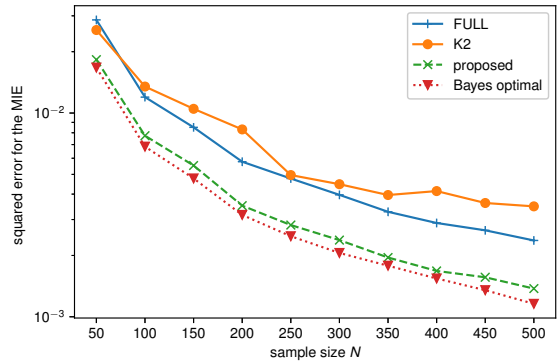}}
\caption{Curves of the squared errors for the MIE for a small $G_{\rm full}$ as the function  of sample size. The means of the squared errors of 1000 experiments are depicted.}
\label{fig:experiment_Bayes}
\end{center}
\vskip -0.2in
\end{figure}

\begin{figure*}[t]
\vskip 0.2in
\begin{center}
\centerline{\includegraphics[width=1.\linewidth]{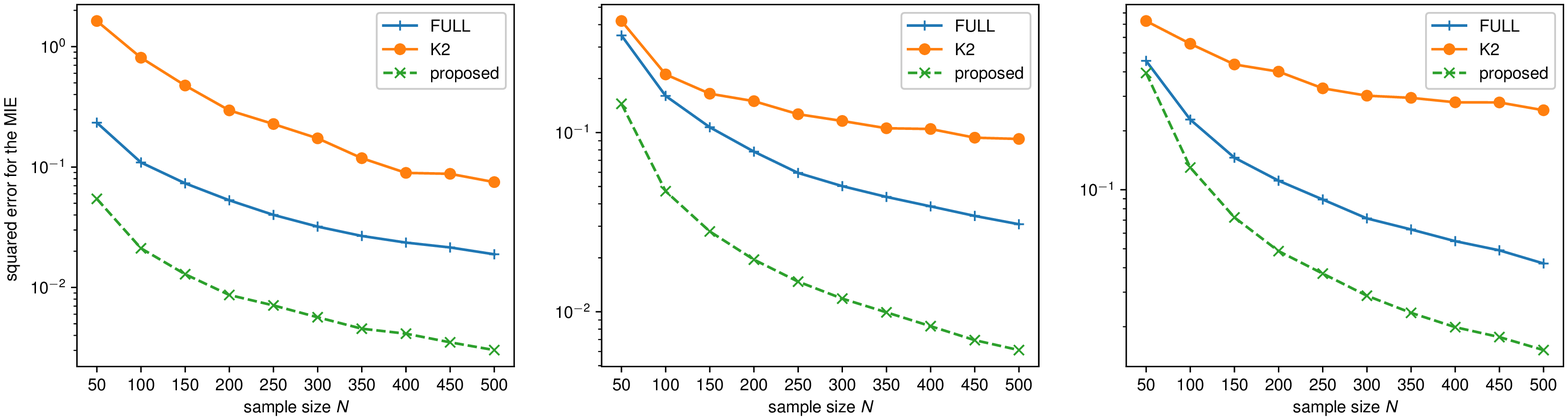}}
\caption{Curves of the squared errors for the MIE for $G_{\rm full}$ with $n_{1}=n_{2}=30$ (left: $p=0.3$, center: $p=0.5$, right: $p=0.7$). The means of the squared errors of 1000 experiments are depicted.}
\label{fig:experiment_2}
\end{center}
\vskip -0.2in
\end{figure*}

\begin{figure}[t]
\vskip 0.2in
\begin{center}
\centerline{\includegraphics[width=0.8\columnwidth]{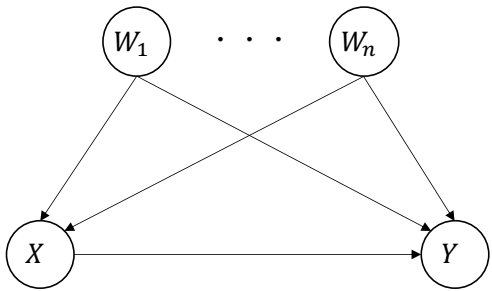}}
\caption{Assumed $G_{\rm full}$ for the experiments on semi-synthetic data.}
\label{fig:graph_for_semi_synthetic}
\end{center}
\vskip -0.2in
\end{figure}

\begin{figure*}[t]
\vskip 0.2in
\begin{center}
\centerline{\includegraphics[width=1.\linewidth]{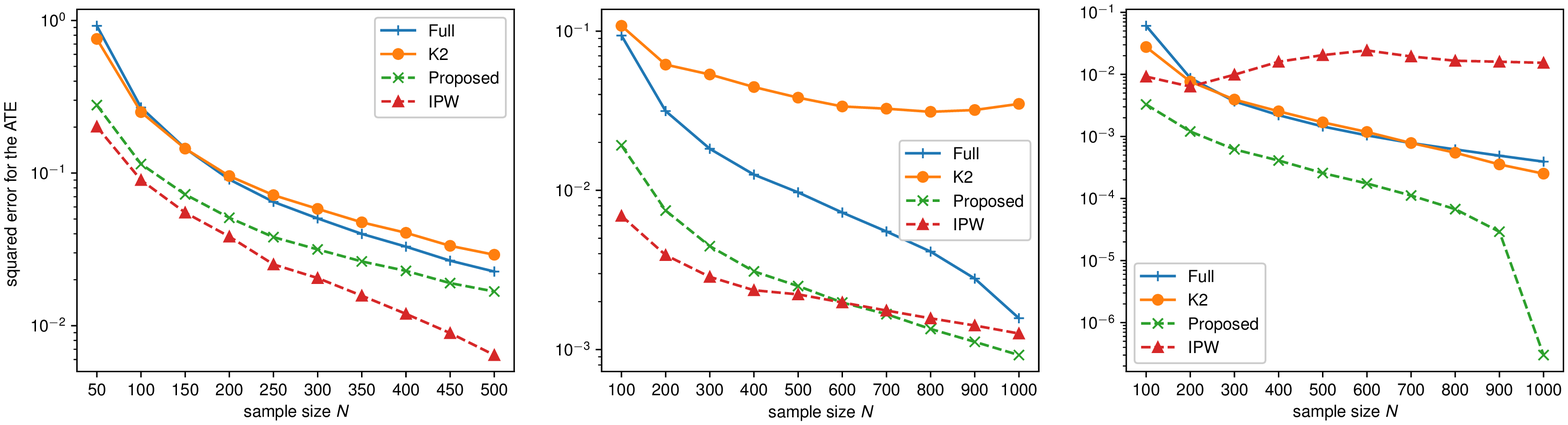}}
\caption{Curves of the squared errors for the MIE for semi-synthetic data (left: IHDP, center and right: LBIDD). The means of the squared errors of 1000 experiments are depicted.}
\label{fig:experiment_semi_synthetic}
\end{center}
\vskip -0.2in
\end{figure*}

\subsection{Experiments on semi-synthetic data}
In real-world applications, ground truth causal effects are rarely available.
Thus, we evaluate our proposed method through some experiments on semi-synthetic data.
For the evaluation, we use two pre-established benchmarking data for causal inference.
\begin{itemize}
\item \textbf{IHDP: } The dataset is constructed from the Infant Health and Development Program (IHDP) \cite{hill2011bayesian}. 
Each observation consists of 25 covariates, an indicator variable that indicates whether the infant received a special care, and an outcome variable that represents the final cognitive test score.
The dataset has 747 observations.

\item \textbf{LBIDD: } The dataset was developed for the 2018 Atlantic Causal Inference Conference competition \cite{shimoni2018benchmarking}. It was derived from the Linked Birth and Infant Death Data (LBIDD).
There are 63 distinct data which are generated from different distributions and we randomly pick two of them.
Each observation consists of 177 covariates, an indicator variable for the treatment status, and an outcome variable.
The dataset has 1000 observations.

\end{itemize}

Both IHDP and LBIDD data consists of some covariates, an indicator variable for the treatment status, and an outcome variable.
We denote the covariates as $W_{1},\ldots, W_{n}$, the indicator variable for the treatment status as $X$, and the outcome variable as $Y$.
We are interested in estimating $\bar{y}_{1}-\bar{y}_{0}$, namely, the average treatment effect (ATE) of $X$ on $Y$.
Since both datasets contain counterfactual data, we can calculate the treatment effect for each record.
We assume that the average of them is the true value of the ATE and evaluate the estimators with the squared errors between the true ATE and estimates.

For both datasets, we assume that $G_{\rm full}$ is the form of Figure \ref{fig:graph_for_semi_synthetic}.
That is, we assume that $W_{1},\ldots,W_{n}$ are the possible common causes between $X$ and $Y$.
The hyperparameters for each estimation algorithm are set to the same values as those for the experiments on synthetic data.

We also compare the methods with the IPW estimator, a well known non-Bayesian estimator of the ATE \cite{horvitz1952generalization}.
IPW estimator requires an estimator of the conditional probability $p(x|w_{1},\ldots,w_{n})$.
We modeled the conditional probability by logistic regression model and estimate its parameters by $\ell_{1}$ penalized maximum likelihood estimation.

Figure \ref{fig:experiment_semi_synthetic} shows the mean squared curves for the ATE.
We can see that the proposed method is superior to the conventional methods.
Especially, for the first LBIDD data, the performance of the K2 algorithm based method is poor and it seems to fail to capture a good causal model.
Comparering the proposed method with the IPW estimator, the IPW estimator has a better estimation accuracy than the proposed estimator for IHDP data; the IPW and proposed estimators are competitive for the first LBIDD data; the proposed estimator outperforms IPW estimator for another LBIDD data.
The experimental results imply that our proposed method is robust to misspecification of the causal model.
We believe that this robustness comes from the fact that the proposed method is build on the idea of weighting the estimates computed under multiple models.

\section{Conclusion}

We proposed a Bayes optimal estimator of the MIE which minimizes the squared error loss on average when the data generating model is an unknown random variable.
The proposed estimator has to estimate the causal effects under the all candidate causal graphs and it is hard to compute when the number of candidate causal graphs is large.
We also proposed an approximation algorithm for the optimal estimator by using a sparse modeling technique.
Some numerical experiments corroborated the effectiveness of our proposed methods.
The proposed methods can help analysts when there is some ambiguity in the knowledge of the data generating model.

We have made some strong assumptions in this study:
\begin{itemize}
    \item Gaussian assumption on the prior or the error terms
    \item Assumption that there is no hidden confounders
    \item Assumption that causal orders are known
    \item Linear assumption on the SCM
\end{itemize}
These assumptions might limit the range of the real-world applications of the proposed methods.
We think that we can relax some assumptions by combining the proposed methods with previously known results.
For example, if we can assume non-Gaussian distributions on the error terms, we can relax the assumption on causal orders since models with opposite causal directions would have different posterior probabilities.
In such cases, the method proposed in this paper needs to be modified, and the constrution of such an algorithm would be an attractive research direction.

\newpage

\ 

\newpage

\section*{Acknowledgments}

This research is partially supported by Information Services International-Dentsu (ISID), Ltd. Research Grant, the Kayamori Foundation of Informational Science Advancement, and No. 19K12128 of Grant-in-Aid for Scientific Research Category (C) and No. 18H03642 of Grant-in Aid for Scientific Research Category (A), Japan Society for the Promotion of Science.

\bibliographystyle{apalike}
\bibliography{aistats2021}

\begin{thebibliography}{}

\bibitem[Andrews and Mallows, 1974]{andrews1974scale}
Andrews, D.~F. and Mallows, C.~L. (1974).
\newblock Scale mixtures of normal distributions.
\newblock {\em Journal of the Royal Statistical Society: Series B
  (Methodological)}, 36(1):99--102.

\bibitem[Angrist et~al., 1996]{angrist1996identification}
Angrist, J.~D., Imbens, G.~W., and Rubin, D.~B. (1996).
\newblock Identification of causal effects using instrumental variables.
\newblock {\em Journal of the American statistical Association},
  91(434):444--455.

\bibitem[Babacan et~al., 2014]{babacan2014bayesian}
Babacan, S.~D., Nakajima, S., and Do, M.~N. (2014).
\newblock Bayesian group-sparse modeling and variational inference.
\newblock {\em IEEE transactions on signal processing}, 62(11):2906--2921.

\bibitem[Berger, 2013]{berger2013statistical}
Berger, J.~O. (2013).
\newblock {\em Statistical decision theory and Bayesian analysis}.
\newblock Springer Science \& Business Media.

\bibitem[Bishop, 2006]{bishop2006pattern}
Bishop, C.~M. (2006).
\newblock {\em Pattern recognition and machine learning}.
\newblock springer.

\bibitem[Boris~Choy and Chan, 2008]{boris2008scale}
Boris~Choy, S. and Chan, J.~S. (2008).
\newblock Scale mixtures distributions in statistical modelling.
\newblock {\em Australian \& New Zealand Journal of Statistics},
  50(2):135--146.

\bibitem[Carvalho et~al., 2009]{carvalho2009handling}
Carvalho, C.~M., Polson, N.~G., and Scott, J.~G. (2009).
\newblock Handling sparsity via the horseshoe.
\newblock In {\em Artificial Intelligence and Statistics}, pages 73--80.

\bibitem[Cooper and Herskovits, 1992]{cooper1992bayesian}
Cooper, G.~F. and Herskovits, E. (1992).
\newblock A bayesian method for the induction of probabilistic networks from
  data.
\newblock {\em Machine learning}, 9(4):309--347.

\bibitem[Figueiredo, 2003]{figueiredo2003adaptive}
Figueiredo, M.~A. (2003).
\newblock Adaptive sparseness for supervised learning.
\newblock {\em IEEE transactions on pattern analysis and machine intelligence},
  25(9):1150--1159.

\bibitem[Fisher, 1951]{fisher1951design}
Fisher, R.~A. (1951).
\newblock The design of experiments.

\bibitem[Guo and Fraser, 2014]{guo2014propensity}
Guo, S. and Fraser, M.~W. (2014).
\newblock {\em Propensity score analysis: Statistical methods and
  applications}, volume~11.
\newblock SAGE publications.

\bibitem[Heckerman et~al., 1995]{heckerman1995learning}
Heckerman, D., Geiger, D., and Chickering, D.~M. (1995).
\newblock Learning bayesian networks: The combination of knowledge and
  statistical data.
\newblock {\em Machine learning}, 20(3):197--243.

\bibitem[Hill, 2011]{hill2011bayesian}
Hill, J.~L. (2011).
\newblock Bayesian nonparametric modeling for causal inference.
\newblock {\em Journal of Computational and Graphical Statistics},
  20(1):217--240.

\bibitem[Hoeting et~al., 1999]{hoeting1999bayesian}
Hoeting, J.~A., Madigan, D., Raftery, A.~E., and Volinsky, C.~T. (1999).
\newblock Bayesian model averaging: a tutorial.
\newblock {\em Statistical science}, pages 382--401.

\bibitem[Horii and Suko, 2019]{horii2019note}
Horii, S. and Suko, T. (2019).
\newblock A note on the estimation method of intervention effects based on
  statistical decision theory.
\newblock In {\em 2019 53rd Annual Conference on Information Sciences and
  Systems (CISS)}, pages 1--6. IEEE.

\bibitem[Horvitz and Thompson, 1952]{horvitz1952generalization}
Horvitz, D.~G. and Thompson, D.~J. (1952).
\newblock A generalization of sampling without replacement from a finite
  universe.
\newblock {\em Journal of the American statistical Association},
  47(260):663--685.

\bibitem[Ishwaran et~al., 2005]{ishwaran2005spike}
Ishwaran, H., Rao, J.~S., et~al. (2005).
\newblock Spike and slab variable selection: frequentist and bayesian
  strategies.
\newblock {\em The Annals of Statistics}, 33(2):730--773.

\bibitem[Ji et~al., 2008]{ji2008bayesian}
Ji, S., Xue, Y., and Carin, L. (2008).
\newblock Bayesian compressive sensing.
\newblock {\em IEEE Transactions on signal processing}, 56(6):2346--2356.

\bibitem[Kaplan and Chen, 2014]{kaplan2014bayesian}
Kaplan, D. and Chen, J. (2014).
\newblock Bayesian model averaging for propensity score analysis.
\newblock {\em Multivariate behavioral research}, 49(6):505--517.

\bibitem[Le~Cam, 2012]{le2012asymptotic}
Le~Cam, L. (2012).
\newblock {\em Asymptotic methods in statistical decision theory}.
\newblock Springer Science \& Business Media.

\bibitem[Park and Casella, 2008]{park2008bayesian}
Park, T. and Casella, G. (2008).
\newblock The bayesian lasso.
\newblock {\em Journal of the American Statistical Association},
  103(482):681--686.

\bibitem[Pearl, 2000]{pearl2000causality}
Pearl, J. (2000).
\newblock Causality: Models, reasoning, and inference.

\bibitem[Rubin, 1978]{rubin1978bayesian}
Rubin, D.~B. (1978).
\newblock Bayesian inference for causal effects: The role of randomization.
\newblock {\em The Annals of statistics}, pages 34--58.

\bibitem[Shimizu et~al., 2006]{shimizu2006linear}
Shimizu, S., Hoyer, P.~O., Hyv{\"a}rinen, A., and Kerminen, A. (2006).
\newblock A linear non-gaussian acyclic model for causal discovery.
\newblock {\em Journal of Machine Learning Research}, 7(Oct):2003--2030.

\bibitem[Shimoni et~al., 2018]{shimoni2018benchmarking}
Shimoni, Y., Yanover, C., Karavani, E., and Goldschmnidt, Y. (2018).
\newblock Benchmarking framework for performance-evaluation of causal inference
  analysis.
\newblock {\em arXiv preprint arXiv:1802.05046}.

\bibitem[Spirtes and Glymour, 1991]{spirtes1991algorithm}
Spirtes, P. and Glymour, C. (1991).
\newblock An algorithm for fast recovery of sparse causal graphs.
\newblock {\em Social science computer review}, 9(1):62--72.

\bibitem[Wermuth and Lauritzen, 1982]{wermuth1982graphical}
Wermuth, N. and Lauritzen, S.~L. (1982).
\newblock {\em Graphical and recursive models for contigency tables}.
\newblock Institut for Elektroniske Systemer, Aalborg Universitetscenter.

\bibitem[Wright, 1921]{wright1921correlation}
Wright, S. (1921).
\newblock Correlation and causation.
\newblock {\em J. agric. Res.}, 20:557--580.

\end{thebibliography}

\onecolumn
\appendix

\section{Derivation of the analytical form of $p(G|D^{N})$ and $p(\bm{\theta}_{G}|G, D^{N})$}
First, we derive $p(\bm{\theta}_{G}|G,D^{N})$ for a fixed $G\in\mathcal{G}$.
For $j\in\left\{1,\ldots,m\right\}$, let ${\rm pa}(X_{j})=(X_{j_{1}},X_{j_{2}},\ldots,X_{j_{m_{j}}})$ and $\bm{X}_{j}=[\bm{x}_{j_{1}},\bm{x}_{j_{2}},\ldots,\bm{x}_{j_{m_{j}}}]\in \mathbb{R}^{N\times m_{j}}$, where $\bm{x}_{i}\in\mathbb{R}^{N}$ is the sample of $X_{i}$.
Then, for $\bm{\theta}_{j}=(\theta_{j_{1}j},\theta_{j_{2}j},\ldots,\theta_{j_{m_{j}}j})$, the likelihood function $p(D^{N}|G,\bm{\theta}_{j})$ is given by
\begin{align}
    p(D^{N}|G,\bm{\theta}_{j})=\mathcal{N}(\bm{x}_{j};\bm{X}_{j}\bm{\theta}_{j},\tau\bm{I}_{m_{j}})+{\rm const.},
\end{align}
where $\bm{I}_{m_{j}}$ is the identity matrix of size $m_{j}$.
Since we assumed a conjugate Gaussian prior for $p(\bm{\theta}_{G}|D)$, the posterior distribution $p(\bm{\theta}_{j}|G,D^{N})$ is given by
\begin{align}
   p(\bm{\theta}_{j}|G,D^{N})&=\mathcal{N}(\bm{\theta}_{j};\bm{\mu}_{j},\bm{\Sigma}_{j}),\label{parameter_posterior}\\
   \bm{\mu}_{j}&=s_{\epsilon}\bm{\Sigma}_{j}\bm{X}_{j}^{T}\bm{x}_{j},\label{posterior_mean}\\
   \bm{\Sigma}_{j}&=\left(s_{\epsilon}\bm{X}_{j}^{T}\bm{X}_{j}+\tau^{-1}\bm{I}_{m_{j}}\right)^{-1}.
\end{align}

Further, we can calculate the likelihood $p(D^{N}|G)$ as follows.
\begin{align}
    p(D^{N}|G)&=\prod_{j=1}^{m}p(\bm{x}_{j}|\bm{X}_{j}),\\
    p(\bm{x}_{j}|\bm{X}_{j})&=\frac{m_{j}}{2}\ln \tau^{-1}+\frac{N}{2}\ln s_{\epsilon}-E_{j}
    -\frac{1}{2}\ln|\bm{A}_{j}|-\frac{N}{2}\ln(2\pi),\label{model_likelihood}\\
    E_{j}&=\frac{s_{\epsilon}}{2}||\bm{x}_{j}-\bm{X}_{j}\bm{\mu}_{j}||^{2}+\frac{\tau^{-1}}{2}\bm{\mu}_{j}^{T}\bm{\mu}_{j},\\
    \bm{A}_{j}&=\tau^{-1}\bm{I}_{m_{j}}+s_{\epsilon}\bm{X}_{j}^{T}\bm{X}_{j}.
\end{align}
We can calculate the posterior probability $p(G|D^{N})$ by using the Bayes rule.
See \cite{bishop2006pattern} for the derivation of (\ref{parameter_posterior}) and (\ref{model_likelihood}).

\section{Derivation of Variational Bayes algorithm}
The joint distribution for $\bm{x}_{j}, \bm{X}_{j}, \bm{\theta}_{j}, \bm{\tau}_{j}, \bm{\alpha}_{j}$ is factorized as 
\begin{align}
    p(\bm{x}_{j},\bm{X}_{j},\bm{\theta}_{j}, \bm{\tau}_{j}, \bm{\alpha}_{j})=
    p(\bm{x}_{j}|\bm{X}_{j}, \bm{\theta}_{j})p(\bm{\theta}_{j}|\bm{\tau}_{j})p(\bm{\tau}_{j}|\bm{\alpha}_{j})p(\bm{\alpha}_{j};\kappa,\nu).
\end{align}
Let $\bm{\xi}=(\bm{\theta}_{j}, \bm{\tau}_{j}, \bm{\alpha}_{j})$.
The variational Bayes method finds an approximation distribution $q(\bm{\xi})$ that approximates $p(\bm{\xi}|\bm{x}_{j},\bm{X}_{j})$.
The goal is to find $q(\bm{\xi})$ that minimizes the Kullback-Leibler divergence ${\rm KL}(q(\bm{\xi})||p(\bm{\xi}|\bm{x}_{j},\bm{X}_{j}))$:
\begin{align}
    q^{*}(\bm{\xi})&=\argmin_{q(\bm{\xi})}\int q(\bm{\xi})\ln \frac{q(\bm{\xi})}{p(\bm{\xi}|\bm{x}_{j},\bm{X}_{j})}{\rm d}\bm{\xi}\\
    &=\argmin_{q(\bm{\xi})}\int q(\bm{\xi})\ln \frac{q(\bm{\xi})}{p(\bm{\xi},\bm{x}_{j},\bm{X}_{j})}{\rm d}\bm{\xi}.\label{KL}
\end{align}
However, it is difficult to minimize (\ref{KL}) for arbitrary distributions.
We limit the optimization distributions to $q(\bm{\xi})$ that can be factorized as \begin{align}
    q(\bm{\theta}_{j},\bm{\tau}_{j},\bm{\alpha}_{j})=q(\bm{\theta}_{j})q(\bm{\tau}_{j})q(\bm{\alpha}_{j}).
\end{align}
For $\bm{\xi}_{k}\in\bm{\xi}$, the variational Bayes method minimizes (\ref{KL}) by updating $q(\bm{\xi}_{k})$ sequentially.
With the distribution $q(\bm{\xi}\setminus \bm{\xi}_{k})$ of $\bm{\xi}\setminus \bm{\xi}_{k}$ fixed, the update equation of $q(\bm{\xi}_{k})$ is given as follows \cite{bishop2006pattern}.
\begin{align}
    \ln q^{*}(\bm{\xi}_{k})={\rm E}_{q(\bm{\xi}\setminus \bm{\xi}_{k})}\left[\ln p(\bm{\xi},\bm{x}_{j},\bm{X}_{j})\right]+{\rm const.}\label{VB_update}
\end{align}
In the following, we describe concrete update equation of each $q(\bm{\xi}_{k})$.
To keep the description concise, for functions $f(\bm{\xi}_{k})$, the expectation taken by $q(\bm{\xi}_{k})$ at the point is written as $\left<f(\bm{\xi}_{k})\right>$.

\subsection*{Update equation of $q(\bm{\theta}_{j})$}
From (\ref{VB_update}), the update equation of $q(\bm{\theta}_{j})$ is 
\begin{align}
    \ln q^{*}(\bm{\theta}_{j})={\rm E}_{q(\bm{\tau}_{j})}\left[p(\bm{x}_{j}|\bm{X}_{j},\bm{\theta}_{j})p(\bm{\theta}_{j}|\bm{\tau}_{j})\right]+\mbox{const.}
\end{align}
Using the assumption that $p(\bm{x}_{j}|\bm{X}_{j},\bm{\theta}_{j})$ and $p(\bm{\theta}_{j}|\bm{\tau}_{j})$ are Gaussian distributions, we obtain
\begin{align}
    q^{*}(\bm{\theta}_{j})&=\mathcal{N}(\bar{\bm{\theta}}_{j},\tilde{\bm{\Sigma}}_{j}),\\
    \bar{\bm{\theta}}_{j}&=s_{\epsilon}\tilde{\bm{\Sigma}}_{j}\bm{X}_{j}^{T}\bm{x}_{j},\\
    \tilde{\bm{\Sigma}}_{j}&=\left(s_{\epsilon}\bm{X}_{j}^{T}\bm{X}_{j}+\left<\bm{S}_{\bm{\tau}_{j}}\right>\right)^{-1},
\end{align}
where
\begin{align}
    \bm{S}_{\bm{\tau}_{j}}={\rm diag}\left(\tau_{j,1}^{-1},\ldots,\tau^{-1}_{j,m_{j}}\right).
\end{align}

\subsection*{Update equation of $q(\bm{\tau})$}
From (\ref{VB_update}), the update equation of $q(\bm{\tau})$ is 
\begin{align}
    \ln q^{*}(\bm{\tau})={\rm E}_{q(\bm{\theta}_{j},\bm{\alpha}_{j})}\left[p(\bm{x}_{j}|\bm{X}_{j},\bm{\theta}_{j})p(\bm{\theta}_{j}|\bm{\tau}_{j})p(\bm{\theta}_{j}|\bm{\alpha}_{j})\right]+\mbox{const.}\label{tau_update}
\end{align}
From the model assumption, without loss of generality, we can assume that $q(\bm{\tau}_{j})$ is decomposed as 
\begin{align}
    q(\bm{\tau}_{j})=\prod_{i=1}^{m_{j}}q(\tau_{j,i}).
\end{align}
By arranging the terms in (\ref{tau_update}) that include $\tau_{j,i}$, we obtain
\begin{align}
    q^{*}(\tau_{j,i})=\mathcal{GIG}\left(\left<\alpha_{j,i}\right>,\left<\theta_{j,i}^{2}\right>,\frac{1}{2}\right),
\end{align}
where $\mathcal{GIG}(a,b,\rho)$ denotes the generalized inverse Gaussian distribution, whose probability density function is given by
\begin{align}
    p(x;a,b,\rho)=\frac{(a/b)^{\rho/2}}{2K_{\rho}(\sqrt{ab})}x^{\rho-1}\exp\left(-\frac{ax+bx^{-1}}{2}\right),
\end{align}
where $K_{\rho}$ is a modified Bessel function of the second kind.
To update $q(\bm{\theta}_{j})$ and $q(\bm{\alpha}_{j})$, we need the expected values $\left<\tau_{j,i}\right>$ and $\left<\tau_{j,i}^{-1}\right>$.
They are given by
\begin{align}
    \left<\tau_{j,i}\right>&=\frac{1+\sqrt{\left<\tau_{j,i}\right>\left<\theta_{j,i}^{2}\right>}}{\alpha_{j,i}},\\
    \left<\tau_{j,i}^{-1}\right>&=\sqrt{\frac{\left<\alpha_{j,i}\right>}{\left<\theta_{j,i}^{2}\right>}}.
\end{align}

\subsection*{Update equation of $q(\bm{\alpha})$}
From (\ref{VB_update}), the update equation of $q(\bm{\alpha})$ is 
\begin{align}
    \ln q^{*}(\bm{\alpha})={\rm E}_{q(\bm{\tau})}\left[p(\bm{\tau}|\bm{\alpha})p(\bm{\alpha};\kappa,\nu)\right]+\mbox{const.}\label{update_alpha}
\end{align}
As in the case for $\bm{\tau}_{j}$, we can assume that $q(\bm{\alpha}_{j})$ is decomposed as 
\begin{align}
    q(\bm{\alpha}_{j})=\prod_{i=1}^{m_{j}}q(\alpha_{j,i}).
\end{align}
By arranging the terms in (\ref{update_alpha}) that include $\alpha_{j,i}$, we obtain
\begin{align}
    q^{*}(\alpha_{j,i})=\mathcal{GA}\left(\kappa+1,\nu+\frac{\left<\tau_{j,i}\right>}{2}\right),
\end{align}
where $\mathcal{GA}(\kappa,\nu)$ is the gamma distribution, whose probability density function is given by
\begin{align}
     p(x;\kappa,\nu)=\frac{\nu^{\kappa}}{\Gamma(\kappa)}x^{\kappa-1}e^{-\nu x}.
\end{align}
To update $q(\bm{\tau})$, we need the expected value $\left<\alpha_{j,i}\right>$.
It is given by
\begin{align}
    \left<\alpha_{j,i}\right>=\left(\kappa+1\right)\left(\nu+\frac{\left<\tau_{j,i}\right>}{2}\right).
\end{align}

\end{document}